\documentclass[preprint,12pt]{elsarticle}
\usepackage{xcolor}
\usepackage{amssymb}
\usepackage{enumitem}
\usepackage{url}
\usepackage{multirow}
\usepackage{bbm}
\usepackage{booktabs}
\usepackage{multicol}
\usepackage{amsmath}
\usepackage{dsfont}
\usepackage{rotating}
\usepackage{xcolor}
\usepackage{amsthm,}
\usepackage{geometry}
\newtheorem{theorem}{Theorem}[section]
\newtheorem{proposition}[theorem]{Proposition}
\newtheorem{corollary}[theorem]{Corollary}
\newtheorem{lemma}[theorem]{Lemma}

\newtheorem{remark}[theorem]{Remark}

\journal{REVSTAT}

\begin{document}

\begin{frontmatter}

\title{An Extension of the \(d\)-Variate FGM Copula with Application}

\author[label1]{Mous-Abou Hamadou\corref{cor1}}
\ead{mhamadou@quinnipiac.edu}

\author[label2]{Martial Longla}
\ead{longla@olemiss.edu}

\affiliation[label1]{organization={Department of Mathematics and Statistics, Quinnipiac University},
            city={Hamden},
            state={CT},
            postcode={06518},
            country={USA}}

\affiliation[label2]{organization={Department of Mathematics, The University of Mississippi},
            city={Oxford},
            state={MS},
            postcode={38655},
            country={USA}}

\cortext[cor1]{Corresponding author.}
\begin{abstract}
 We introduce an extended \(d\)-variate Farlie–Gumbel–Morgenstern (FGM) copula that incorporates additional parameters based on Legendre polynomials to enhance the representation of multivariate dependence structures. Within an i.i.d. framework, we derive closed-form estimators for these parameters and establish their unbiasedness, consistency, and asymptotic normality. A simulation study illustrates the finite-sample performance of the estimators. The model is applied to the \texttt{Bearing dataset}, previously studied by Ota and Kimura (2021) \cite{OtaKimura2021} through a  \(d\)-variate FGM copula and by  Longla and Mous-Abou (2025)\cite{LonglaMousAbou2025} using an extended bivariate FGM copula. Our analysis shows that the classical \(d\)-variate FGM copula does not adequately represent the dependence in this dataset. Based on estimation results and model selection criteria, we propose a reduced version of the extended model as a more appropriate copula specification for the \texttt{Bearing data.}

\end{abstract}

\end{frontmatter}

\section{Introduction}\label{intro}
Introduced by Sklar (1959) \cite{Sklar1959}, copulas are multivariate cumulative distribution functions with uniform marginals on the unit interval. They have become indispensable tools for capturing and modeling dependence structures among multivariate random variables, independently of their marginal distributions. Recent applications of copulas in fields such as economics, finance, risk management, and reliability engineering can be found in Long and Emura (2014) \cite{LongEmura2014}, Sun et al. (2020) \cite{Sun2020}, Ota and Kimura (2021) \cite{OtaKimura2021}, and Longla and Mous-Abou (2025) \cite{LonglaMousAbou2025}. Additional references on copula applications can be found in Bhati and Do (2019) \cite{BhatiDo2019}, whereas a thorough introduction and detailed theoretical background on copulas is provided by Nelsen (2006) \cite{Nelsen} and Joe (2014) \cite{Joe2014}.

A well-known parametric family of copulas is the FGM family, originally introduced through contributions by Eyraud (1936) \cite{Eyraud1936}, Morgenstern (1956) \cite{morgenstern1956}, Farlie (1960) \cite{farlie1960performance}, and Gumbel (1960) \cite{gumbel1960bivariate}. Copulas from this family are defined as polynomial perturbations of the independence copula. This polynomial structure facilitates closed-form expressions for key analytical quantities, including the density function, moments, and dependence measures, making the FGM family particularly appealing for practical applications. However, copulas from this family have a narrow dependence range, with Spearman’s Rho limited to \(\bigl[-1/3, 1/3]\) making it difficult to captuire complex dependencies. Various extensions have been proposed to overcome this limitation, including but not limited to those by Huang and Kotz (1999) \cite{HuangKotz1999}, Lai and Xie (2000) \cite{LaiXie2000}, Bairamov and Kotz (2002) \cite{BairamovKotz2002}, Rodríguez-Lallena (2004) \cite{RodriguezLallena2004}, Amblard and Girard (2009) \cite{AmblardGirard2009}, Durante et al. (2013) \cite{Durante2013}, Komorník et al. (2017) \cite{Komornik2017}, Côté and Genest (2019) \cite{CoteGenest2019} and  Bekrizadeh (2022) \cite{Bekrizadeh2022}. More recently, Longla (2024) \cite{Long2024} introduced an extension of the bivariate FGM copula by incorporating an additional parameter, thereby extending the dependence range to \(\bigl[-1/5, 2/5\bigr]\). 

Building on this, the present paper proposes an extension of the \(d\)-variate FGM copula. Originally examined by Joe and Xu (1996) \cite{JoeXu1996} for parameter estimation, the \(d\)-variate FGM copula was further studied by Ota and Kimura (2021) \cite{OtaKimura2021}, who developed an  estimation algorithm using the Inference Function for Margins (IFM) method. Their approach proved to be effective and serves as a practical alternative to MLE, which is nearly infeasible in this context due to the complexity of the parameter constraints and the number of parameters. They demonstrated the utility of their method by applying it to the \texttt{Bearing dataset} (see J. Lee et al. (2007) \cite{Lee2007}).

However, Longla and Mous-Abou (2025) \cite{LonglaMousAbou2025} highlighted limitations of the classical \(d\)-variate FGM copula when applied to the \texttt{Bearing dataset}. In the absence of a multivariate extension, they conducted a pairwise analysis using the extended bivariate FGM copula of Longla (2024) \cite{Long2024}. This approach aimed to assess the significance of additional parameters introduced in the extended form. Their results demonstrated that the majority of the additional parameters were statistically significantly different from zero, thereby indicating that the classical 
\(d\)-variate FGM copula fails to adequately capture the dependence structure present in the dataset. Nevertheless, their study did not incorporate a fully multivariate copula model, nor did it offer a complete alternative to address the identified shortcomings.

The proposed extended \(d\)-variate FGM copula introduces additional parameters that account for more complex dependence patterns. This multivariate formulation enables a more comprehensive assessment of model fit and facilitates the selection of appropriate copula specifications tailored to the dependence structure observed in the \texttt{Bearing dataset}.
\subsection{Definitions, Conventions and Notations}\label{definitionsandnotations}
A \(d\)-variate copula is a function \(C : [0,1]^d \to [0,1]\) satisfying:\begin{enumerate}
\item \(
C(u_1, \dots, u_d) =\left\{\begin{array}{lr}
0 & \text{if } u_j = 0 \text{ for some } j=1,\cdots,d,
 \\
 u_j & \text{if }u_k = 1 \text{ for all } k \neq j, k=1,\cdots, d.
\end{array}\right.\)
\item \(C\) is  \(d\)-increasing on \([0,1]^d\). In other terms , for any hyper-rectangle \(\prod_{j=1}^d [u_{j1}, u_{j2}] \subset [0,1]^d\), where \(0 \le u_{j1} < u_{j2} \le 1\), the C-volume satisfies  
\[
\sum_{i_1=1}^2 \cdots \sum_{i_d=1}^2 (-1)^{i_1 + \cdots + i_d} C(u_{1i_1}, \dots, u_{d i_d}) \ge 0.
\]
\end{enumerate}
The \(d\)-variate FGM copula, denoted by \(C_d^{FGM}\) , is defined as follows (see Nelsen, 2006 \cite{Nelsen}; Johnson and Kotz, 1995 \cite{JohnsonKott1975}):
\begin{equation}\label{classicalFGM}C^{FGM}_d(u_1,\cdots,u_d;\Lambda)=\prod\limits_{j=1}^d u_i\bigg(1+\sum\limits_{k=2}^d\sum\limits_{1\leq j_1<\cdots<j_k\leq d}\lambda_{j_1,\cdots,j_k}\prod\limits_{i=1}^k(1-u_{j_i})\bigg),\end{equation}
where \((u_1,\cdots,u_d)\in[0,1]^d\) and \(\lambda_{j_1,\cdots,j_k}\in \Lambda\) are parameters defining the dependence structure, subject to the constraint:\[1+\sum\limits_{k=2}^d\sum\limits_{1\leq j_1<\cdots< j_k\leq d}\lambda_{j_1\cdots j_k}\varepsilon_{j_1}\cdots\varepsilon_{j_k}\geq 0,\quad \{\varepsilon_{j_1},\dots, \varepsilon_{j_k}\}\in \{-1,1\}^d.\]

Throughout this paper, the set \(D = \{1, 2, \dots, d\}\) denotes the index and for any positive integer \(j\), we write \(\mathcal{M}_j(D)\) to represent the collection of all subsets of \(D\) with cardinality \(j\), and \(|M|\) denotes the cardinality of a set \(M\). Indicator functions are denoted by \(\mathbf{1}_A\), which equals 1 if the condition \(A\) is true and 0 otherwise. To maintain consistency in expressions involving empty products or subsets, we adopt the convention that \(0^0 = 1\). Unless otherwise specified, all vectors are treated as column vectors, and the transpose of a vector \(x\) is denoted by \(x^\top\). We use \(\mathbf{0}\) to denote the zero vector of the appropriate dimension.

\subsection{Structure of the paper}\label{structure}
The remainder of the paper is organized as follows. Section \ref{thedvariateFGMcopula} introduces the extended \(d\)-variate FGM copula and presents its defining properties. Section \ref{paramterestimation} develops closed-form estimators for the model parameters and establishes their asymptotic properties. A simulation study assessing the finite-sample performance of the estimators is presented in Section \ref{simulationstudy}. Section \ref{casestudy} applies the proposed model to the \texttt{Bearing dataset} and compares it to the classical \(d\)-variate FGM copula. 

\section{The \(d\)-Variate Extended FGM Copula}\label{thedvariateFGMcopula}
Consider a random vector \( \textbf{U}=(U_1, \dots, U_d)\) following a \(d\)-variate extended FGM copula with uniform marginals on \([0,1]\).  The joint distribution function of the \(d\)-variate extended FGM copula is given by
\begin{eqnarray}\label{copfunction}
C(u_1, \dots, u_d)&\hspace{-0.3cm} = &\mathbb{P}[U_1\leq u_1,\cdots,U_d\leq u_d]\nonumber\\&\hspace{-0.3cm}=&\prod_{j=1}^d u_j + \sum_{k=1}^2 \sum_{j=2}^d \sum_{M \in \mathcal{M}_j(D)} \lambda^{(k)}_M \, \prod_{m \notin M} u_m \, \prod_{m \in M} \Phi_k(u_m),~~~
\end{eqnarray} \[\text{where }(u_1,\cdots,u_d)\in[0,1]^d, \Phi_1(x)=\sqrt{3}(x^2-x) \text{ and } ~\Phi_2(x)=\sqrt{5}(2x^3-3x^2+x) .\] Moreover, the parameters \(\lambda^{(k)}_M\) for \(k=1,2\)  satisfy the constraint
\[
\sum_{j=2}^d \bigg[\sqrt{3}^{j} \sum_{M \in \mathcal{M}_j(D)}  \,|\lambda^{(1)}_M| + \sqrt{5}^{|j|}\sum_{M \in \mathcal{M}_j(D)}  \, |\lambda^{(2)}_M|\bigg]  \leq 1.
\]
The corresponding joint density is then given by
\begin{equation}\label{ExtendedDensity}
c(u_1, \dots, u_d) = 1 + \sum_{k=1}^2 \sum_{j=2}^d \sum_{M \in \mathcal{M}_j(D)} \lambda^{(k)}_M \, \prod_{m \in M} \phi_k(u_m).
\end{equation} Where  
\[\phi_k(x)=\tfrac{\partial \Phi_k(x)}{\partial x}, \text{ for }~k=1,2.\]

 Each parameter represents a joint moment of the transformed variables, capturing dependence in the direction of the associated orthogonal function over the subset of variables indexed by \(M\). In particular, when the subset has two elements, the parameter equals the correlation between those variables.

The \(d\)-variate extended FGM copula includes \(\sum\limits_{j=2}^d 2\binom{d}{j} = 2^{d+1} - 2d - 2\) parameters. 

For instance, when \(d = 4\), the model involves 22 parameters, and the copula density is given by
\begin{eqnarray}\label{extendedFGMcopulaDensity}
    c(u_1,u_2,u_3,u_4)= 1 + \sum_{k=1}^2 \mathbf{v}^\top_k\Lambda^{(k)} .
\end{eqnarray}
Where  for each \(k\), 
\begin{eqnarray}\label{Lambdak}
\Lambda^{(k)}&=&(
\lambda^{(k)}_{\{1,2\}},
\lambda^{(k)}_{\{1,3\}},
\lambda^{(k)}_{\{1,4\}},
\lambda^{(k)}_{\{2,3\}},
\lambda^{(k)}_{\{2,4\}},
\lambda^{(k)}_{\{3,4\}},
\lambda^{(k)}_{\{1,2,3\}},
\lambda^{(k)}_{\{1,2,4\}},
\lambda^{(k)}_{\{1,3,4\}},\nonumber\\&&
\lambda^{(k)}_{\{2,3,4\}},
\lambda^{(k)}_{\{1,2,3,4\}})^\top,
\end{eqnarray}
and 
\begin{eqnarray}\label{vk}
\mathbf{v}_k &=& \Big( \phi_k(u_1)\phi_k(u_2),\; \phi_k(u_1)\phi_k(u_3),\; \phi_k(u_1)\phi_k(u_4),\; \phi_k(u_2)\phi_k(u_3), \nonumber\\
&&\quad \phi_k(u_2)\phi_k(u_4),\; \phi_k(u_3)\phi_k(u_4),\; \phi_k(u_1)\phi_k(u_2)\phi_k(u_3),\; \phi_k(u_1)\phi_k(u_2)\phi_k(u_4), \nonumber\\
&&\quad \phi_k(u_1)\phi_k(u_3)\phi_k(u_4),\; \phi_k(u_2)\phi_k(u_3)\phi_k(u_4),\; \phi_k(u_1)\phi_k(u_2)\phi_k(u_3)\phi_k(u_4) \Big)^\top.\nonumber\\
\end{eqnarray}
\begin{remark}\label{subvector}  
If \((U_1, \dots, U_d)\) follows a copula with density given by \eqref{ExtendedDensity} and \(P\subseteq D\) with \(|P|\geq 2\), then the density of the distribution of the subvector \({\bf U}=(U_{s})_{s\in P}\) is given by  
\begin{equation}  
\label{densitysubvector}  
c({\bf u}) = 1 + \sum_{k=1}^2 \sum_{j=2}^{|P|} \sum_{M \in \mathcal{M}_j(P)} \lambda^{(k)}_M \prod_{m \in M} \phi_k(u_m).  
\end{equation}
\end{remark}  
This follows by first obtaining the marginal distribution of the subvector, which is done by setting the unused components to one. Since \(\Phi_k(1) = 0\), all terms involving those components vanish. The resulting expression is then differentiated with respect to the remaining variables to yield the marginal density.

\section{ Parameter Estimation and Asymptotic Distribution}\label{paramterestimation}
This section develops estimators for the parameters of the 
\(d\)-variate extended FGM copula in \eqref{ExtendedDensity} and establishes their asymptotic properties. For \(k, r, z \in \{1, 2\}\), we define

\[\mathcal{I}_{r,z}=\int_0^1 \phi_r(u) \phi_z(u) \, du \quad and \quad \mathcal{I}_{k, r, z} = \int_0^1 \phi_k(u) \phi_r(u) \phi_z(u) \, du.\]

\begin{lemma}\label{lemma1} Due to the orthogonality of the set \(\{1,\phi_1(x),\phi_2(x)\}\) for \(x\in[0,1]\), it's easy to see that \(\mathcal{I}_{r,z}=
\mathbf{1}_{\{r=z\}}\) and 
\begin{equation}\label{integrals} 
\mathcal{I}_{k, r, z} =
\begin{cases} 
\tfrac{2\sqrt{5}}{7}, & k = r = z = 2, \\
\tfrac{2}{\sqrt{5}}, &(k, r, z) \in \{(1,1,2), (1,2,1), (2,1,1)\}, \\
0, & \text{otherwise}.
\end{cases}
\end{equation}
\end{lemma}
~\\

Suppose that \((U_1, \dots, U_d)\) follow a copula with density given by \eqref{ExtendedDensity}. Let \( P, Q \subseteq D \) be two subsets such that \( 2 \leq |P| \leq |Q| \), and define  
\begin{equation}\label{E}  
E^{z,r}_{P,Q} = \mathbb{E} \bigg[ \prod_{s \in P} \phi_z(U_s) \prod_{t \in Q} \phi_r(U_t) \bigg].  
\end{equation}  
 We establish the following result. 
\begin{proposition}\label{calculus} For \((U_1, \ldots, U_d)\) following a copula with density given in \eqref{ExtendedDensity}, and for sets \( P \) and \( Q \), with \( E^{z,r}_{P,Q} \) as defined in \eqref{E}. The following holds:  
\begin{enumerate}
\item 
\begin{equation}\label{1stequation}
\mathbb{E}\big[\prod\limits_{s \in P}\phi_z(U_s)\big]
=\lambda^{(z)}_P
.\end{equation}   
\item If \(P \subseteq Q\),
\begin{eqnarray}\label{2ndequation}  
E^{z,r}_{P,Q}&=&\mathcal{I}_{z,r}\mathbf{1}_{\{P=Q\}}+\sum\limits_{k=1}^2\sum\limits_{j=2}^{\lvert Q\lvert} \sum\limits_{M \in \mathcal{M}_j(Q)} \lambda^{(k)}_{M}\mathcal{I}_{k,r,z}^{|M\cap P|}\mathcal{I}_{k,r}^{|(Q \setminus P)\cap M|}.
\end{eqnarray}
\item If \(P\not\subset Q\), \begin{eqnarray}\label{3rdequation}
E^{z,r}_{P,Q}=\sum\limits_{k=1}^2\sum\limits_{j=2}^{\lvert P\cup Q\lvert} \sum\limits_{M \in \mathcal{M}_j(P \cup Q)} \lambda^{(k)}_{M}\big[\mathcal{I}_{k,r,z}^{|M\cap P \cap Q|}\mathcal{I}_{k,z}^{|(P \setminus Q)\cap M|}\mathcal{I}_{k,r}^{|(Q \setminus P)\cap M|}.\end{eqnarray}
\end{enumerate}
\end{proposition}
\begin{proof} Using the density given by formula \eqref{densitysubvector} in  Remark \ref{subvector}, we first have
\begin{eqnarray*}
\mathbb{E}\Big[\prod_{s \in P}\phi_z(u_s)\Big] &=& [\Phi_z(1)]^{|P|}+\sum_{k=1}^2 \sum_{j=1}^{|P|} \sum_{M \in \mathcal{M}_j(P)} \lambda^{(k)}_{M}\,\mathcal{I}_{k,z}[\Phi_k(1)]^{|P\setminus M|}.
\end{eqnarray*}
Given that \(\Phi_k(1)=0\) for \(k=1,2\), the expression simplifies to
\begin{equation}\label{unbiased}
\mathbb{E}\Bigg[\prod_{s \in P}\phi_z(U_s)\Bigg] = \sum_{k=1}^2 \sum_{j=2}^{|P|} \sum_{M \in \mathcal{M}_j(P)} \lambda^{(k)}_{M}\,\mathcal{I}_{k,z}\,\mathbf{1}_{\{P=M\}}.
\end{equation}
Here, \( \mathcal{I}_{k,z} \) forces \( k = z \), while  \( \mathbf{1}_{\{P=M\}} \) imposes \( M = P \). As a result, only the term \( \lambda^{(z)}_P \) remains in \eqref{unbiased}, thereby establishing \eqref{1stequation}. \\

On the other hand, observe that 
\[\prod_{s\in P}\phi_z(u_s)\prod_{t\in Q}\phi_r(u_t)=\prod_{s\in P\cap Q}\phi_z(u_s)\phi_r(u_s)\prod_{s\in P\setminus Q}\phi_z(u_s)\prod_{s\in Q\setminus P}\phi_r(u_s) \] 
Using the density \eqref{densitysubvector} in Remark \ref{subvector}, it follows that
\begin{eqnarray}\label{generalexpression}
E^{z,r}_{P,Q}&=&\big(\int_0^1\phi_z(u)\phi_r(u)du\big)^{|P\cap Q|}\big(\int_0^1\phi_z(u)du\big)^{|Q\setminus P|} \big(\int_0^1\phi_r(u)du\big)^{|P\setminus Q|}\nonumber\\&& + \sum_{k=1}^2 \sum_{j=2}^{|P\cup Q|} \sum_{M \in \mathcal{M}_j(P\cup Q)}\big[ \lambda^{(k)}_M \big(\int_0^1\phi_k(u)\phi_z(u)\phi_r(u)du\big)^{|M\cap P\cap Q|}\nonumber\\&&\hspace{-0.3cm}\times\big(\int_0^1\phi_k(u)\phi_z(u)du\big)^{|M\cap(Q\setminus P)|}\big(\int_0^1\phi_k(u)\phi_r(u)du\big)^{|M\cap(P\setminus Q)|}\big].
\end{eqnarray}
Formula \eqref{generalexpression} further simplifies to \begin{eqnarray}\label{generalexpression1}
E^{z,r}_{P,Q}&=&\mathcal{I}_{z,r}^{|P\cap Q|}[\Phi_z(1)]^{|P \setminus Q|}[\Phi_r(1)]^{|Q \setminus P|}~~~~~~~~~~~~~~~~~~~~~~\nonumber\\&+&\hspace{-0.1cm}\sum\limits_{k=1}^2\sum\limits_{j=2}^{\lvert P\cup Q\lvert} \sum\limits_{M \in \mathcal{M}_j(P \cup Q)} \lambda^{(k)}_{M}\big[\mathcal{I}_{k,r,z}^{|M\cap P \cap Q|}\mathcal{I}_{k,z}^{|(P \setminus Q)\cap M|}\mathcal{I}_{k,r}^{|(Q \setminus P)\cap M|}\big].
\end{eqnarray}

If \( P \not\subset Q \), given that \( |P| \leq |Q| \), we have \( P \setminus Q \neq \emptyset \). In this case, \( [\Phi_z(1)]^{|P \setminus Q|} = 0 \), which causes the first term on the right-hand side of formula \eqref{generalexpression1} to vanish. Consequently, equation \eqref{3rdequation} holds.  

On the other hand, if \( P \subseteq Q \), then \( P \setminus Q = \emptyset \). 
It follows that \( [\Phi_z(1)]^{|P\setminus Q|} = 1 \),\quad  \( \mathcal{I}_{k,z}^{|(P \setminus Q)\cap M|} = 1 \), and \( [\Phi_z(1)]^{|Q\setminus P|} = \mathbf{1}_{\{P=Q\}} \). Substituting these into formula \eqref{generalexpression1} directly yields \eqref{2ndequation}, completing the proof.  
\end{proof}
Suppose that $\{(U_{i1}, \ldots, U_{id})\}_{i=1}^n$ is a simple random sample drawn from a copula with density given in \eqref{ExtendedDensity} and with uniform $\mathcal{U}[0,1]$ marginals. For each integer $j$ with $2\leq j\leq d$ and for every subset $M \in \mathcal{M}_j(D)$, define the estimator
\begin{equation}\label{estimate}
\hat{\lambda}^{(k)}_{M} = \tfrac{1}{n} \sum_{i=1}^n \prod_{m\in M} \phi_k(U_{im}), \quad \text{for } k=1,2.
\end{equation}

Next, let \(\Lambda\) be the parameter vector containing all \(\lambda^{(k)}_M\) for \(k=1,2\), \(j=2,\dots,d\), and \(M \in \mathcal{M}_j(D)\), with \(\hat{\Lambda}\) as its estimator. The following theorem establishes the properties of the estimator in \eqref{estimate} and the asymptotic distribution \(\hat\Lambda\).
\begin{theorem}\label{cltlegendreextended}
Let $\{(U_{i1}, \ldots, U_{id})\}_{i=1}^n$ be a simple random sample generated by a copula with density given in \eqref{ExtendedDensity} and uniform $\mathcal{U}[0,1]$ marginals. For each $j=2,\dots,d$ and $M \in \mathcal{M}_j(D)$, the estimator defined in \eqref{estimate} is unbiased and consistent for $\lambda^{(k)}_M$. Moreover,  
\begin{equation}\label{clt}  
\sqrt{n}(\hat{\Lambda}-\Lambda) \Rightarrow N(\mathbf{0}, \Sigma),  
\end{equation}  
where the covariance matrix \(\Sigma\) has entries  
\(\Sigma_{P,Q}^{z,r}\) for \(z,r=1,2\) and \(P \in \mathcal{M}_p(D), Q \in \mathcal{M}_q(D)\) with \(2 \leq p,q \leq d\), given by  
\begin{equation}\label{covariancematrix}  
\Sigma_{P,Q}^{z,r} = E^{z,r}_{P,Q}-\lambda^z_P\lambda_Q^r,  
\end{equation}  
where \(E^{z,r}_{P,Q}\) is defined in Proposition \ref{calculus}.
\end{theorem}

\begin{proof}
Given that  
\(\{(U_{i1}, \ldots, U_{id})\}_{i=1}^n\) is an i.i.d. sequence, applying Proposition \ref{calculus} yields  

\[
\mathbb{E}[\hat{\lambda}^{(k)}_{M}] = \mathbb{E} \Bigg[ \prod_{m \in M} \phi_k(U_{1m}) \Bigg] = \lambda^{(k)}_M,
\]
which establishes the unbiasedness of \( \hat{\lambda}^{(k)}_M \). Moreover, Applying Proposition \ref{calculus}, we have
\begin{eqnarray*}
    \text{var}(\hat\lambda_M^{k})&=&\tfrac{1-{\lambda^{(k)}_M}^2}{n}+\tfrac{1}{n}\sum\limits_{k=1}^2\sum\limits_{j=2}^{\lvert P\lvert} \sum\limits_{M \in \mathcal{M}_j(P)} \lambda^{(k)}_{M}\mathcal{I}_{k,r,z}^{|M|}\rightarrow 0, \text{ as } n\rightarrow \infty. 
\end{eqnarray*}
So \(\hat{\Lambda}\) is consistent as element-wise consistency holds.\\
Let
\(
t = \bigl( t^{(k)}_M : k=1,2; \; M\in \mathcal{M}_j(D), \; j=2,\dots,d \bigr) \in \mathbb{R}^{(2^{d+1}-2d-2)}
\). By the Cramér–Wold device, to prove the convergence in \eqref{clt} it suffices to show that  
\begin{equation}\label{crmm}
\sqrt{n} \big( t' \hat{\Lambda} - t' \Lambda \big) \Rightarrow N(0, t' \Sigma t).
\end{equation} 
Using \eqref{estimate}, we express  
\[
t' \hat{\Lambda} = \tfrac{1}{n}\sum_{i=1}^n X_i,
\text{ where }
X_i = \sum_{k=1}^{2} \sum_{j=2}^{4} \sum_{M \in \mathcal{M}_j(D)} t_{k,M} \prod_{m \in M} \phi_k(U_{im}).
\] Since \( \{X_i\}_{i=1}^{n} \) is an i.i.d. sequence as functions of an i.i.d. sample, applying Proposition \ref{calculus} gives \[\mathbb{E}[X_1]=\sum_{k=1}^{2} \sum_{j=2}^{d} \sum_{M \in \mathcal{M}_j(D)} t_{k,M} \lambda_M^{(k)}=t'\Lambda\quad \text{ and}\]  
\begin{eqnarray*}\text{var}(X_1)&=&\sum_{k,k'=1}^{2} \sum_{j,j'=2}^{d}\sum_{\substack{ M \in \mathcal{M}_j(D)\\ M' \in \mathcal{M}_{j'}(D)}}t^{(k)}_M t^{(k')}_{M'}\text{cov}\big(\prod_{m \in M} \phi_k(U_{1m}), \prod_{m \in M'} \phi_{k'}(U_{1m})\big)\nonumber\\
&=&\sum_{k,k'=1}^{2} \sum_{j,j'=2}^{d}\sum_{\substack{ M \in \mathcal{M}_j(D)\\ M' \in \mathcal{M}_{j'}(D)}}t^{(k)}_M t^{(k')}_{M'}\big(E^{z,r}_{P,Q}-\lambda^z_P\lambda_Q^r\big)\nonumber\\
&=&t'\Sigma t.
\end{eqnarray*}
By the C.L.T for i.i.d random variables, it follows that \[
\sqrt{n}(\Bar{X}-t'\Lambda)\Rightarrow N(0,t'\Sigma t).
\] Thus, formula \eqref{crmm} holds, which in turn implies \eqref{clt}.
\end{proof}

\begin{remark}  \label{CI}
As a consequence of Theorem \ref{cltlegendreextended}, the \((1-\alpha)100\%\) confidence interval for each parameter \(\lambda_M^{(k)}\) of the copula \eqref{copfunction} is given by  
\[
\bigg(\hat\lambda_M^{(k)} - z_{\alpha/2} \sqrt{\tfrac{\hat s^{(k)}_M}{n}}, \hat\lambda_M^{(k)} + z_{\alpha/2} \sqrt{\tfrac{\hat s^{(k)}_M}{n}} \bigg), \quad k = 1,2.
\]  
Here, the variance estimates \(\hat s^{(1)}_M\) and \(\hat s^{(2)}_M\) are given by  
\[
\hat s^{(1)}_M = 1 + \sum\limits_{j=2}^{|M|} \sum\limits_{N\in \mathcal{M}_j(M)} \big(\tfrac{2}{\sqrt{5}}\big)^{|N|} \hat\lambda^{(2)}_N
\text{ and }
\hat s^{(2)}_M = 1 + \sum\limits_{j=2}^{|M|} \sum\limits_{N\in \mathcal{M}_j(M)} \big(\tfrac{2\sqrt{5}}{7}\big)^{|N|} \hat\lambda^{(2)}_N.
\]  
\end{remark}  
In the following corollary, for \( r, z \in \{1,2\} \), we introduce the notation \(\mathds{I}_{ij} = \mathbf{1}_{\{r=i, z=j\}}\) for \( i, j \in \{1,2\} \) and present the components of the covariance matrix \(\Sigma\) in \eqref{covariancematrix} for \( d=4 \).
\begin{corollary} For \( r, z \in \{1,2\} \), \(d=4\), \(P\in\mathcal{M}_p(D)\) and \(Q\in\mathcal{M}_q(D)\), the components of the covariance matrix \(\Sigma\) defined in \eqref{covariancematrix} are such that 
\begin{enumerate}[leftmargin=*,align=left]
\item If \(P\) and \(Q\) are equal:
\begin{align*}
\Sigma_{P,Q}^{z,r} &= \left(\tfrac{4}{5}\mathds{I}_{11}+\tfrac{20}{49}\mathds{I}_{22}\right) \!\!\! \sum_{M\in\mathcal{M}_2(P)}\!\lambda_M^{(2)} 
+ \left(\tfrac{8\sqrt{5}}{25}\mathds{I}_{11}+\tfrac{40\sqrt{5}}{7}\mathds{I}_{22}\right) \!\!\! \sum_{M\in\mathcal{M}_3(P)}\!\lambda_M^{(2)} \\
&\quad + \left(\tfrac{16}{25}\mathds{I}_{11}+\tfrac{400}{2401}\mathds{I}_{22}\right)\lambda_D^{(2)}\mathbf{1}_{\{P=D\}} + \mathcal{I}_{z,r} - (\lambda_P^{z,r})^2\\
&\quad + (1-\mathcal{I}_{z,r})\Bigg[\tfrac{4}{5}\!\!\! \sum_{M\in\mathcal{M}_2(P)}\!\lambda_M^{(1)} 
+ \tfrac{8\sqrt{5}}{25}\!\!\! \sum_{M\in\mathcal{M}_3(P)}\!\lambda_M^{(1)} 
+ \tfrac{16}{25}\lambda_D^{(1)}\mathbf{1}_{\{P=D\}}\Bigg].
\end{align*}

\item If \( P, Q \in \mathcal{M}_2(D) \) are disjoint:  
\begin{align*} 
\Sigma_{P,Q}^{z,r}=(1-\mathcal{I}_{z,r})(
\lambda^{(z)}_P + \lambda^{(r)}_Q)+\mathcal{I}_{z,r}\sum\limits_{j=2}^4 \sum\limits_{M\in\mathcal{M}_j(D)} \lambda^{(z)}_M-\lambda_P^z\lambda_Q^r.
\end{align*}
\item If \( P = \{u, p\}\in \mathcal{M}_2(D) \) and \( Q = \{u, q\} \in \mathcal{M}_2(D)\):
\begin{align*}\hspace{-0.5cm}\Sigma_{P,Q}^{z,r}=\mathcal{I}_{z,r}\lambda_{\{p,q\}}^{(2)}+\tfrac{2\sqrt{5}}{7}\mathds{I}_{22}\sum\limits_{x\in\{p,q,\{p,q\}\}}\lambda_{\{u\}\cup x}^{(2)}+ 
 \tfrac{2}{\sqrt{5}}(\mathds{I}_{12}\lambda_{\{u,q\}}^{(1)}+\mathds{I}_{21}\lambda_{\{u,p\}}^{(1)})-\lambda_P^z\lambda_Q^r.
\end{align*}
\item If  \( P = \{ u, p \}\in \mathcal{M}_2(D) \) and \( Q = \{ u, q_1, q_2 \}\in \mathcal{M}_3(D) \):
\begin{align*}
\Sigma_{P,Q}^{z,r} &= \sum_{x\in\{q_1,q_2,\{q_1,q_2\}\}} \Big( \lambda^{(1)}_{\{p\}\cup x}\mathds{I}_{11} 
+ \tfrac{2}{\sqrt{5}}\lambda^{(1)}_{\{u\}\cup x}\mathds{I}_{12} 
+ \lambda^{(2)}_{\{p\}\cup x}\mathds{I}_{22} \Big) \\
&\quad + \lambda_{\{q_1,q_2\}}^{(1)}(\mathds{I}_{11} + \mathds{I}_{12}) 
+ \tfrac{2\sqrt{5}}{7} \mathds{I}_{22} \sum_{j=2}^4 \sum_{\substack{M\in\mathcal{M}_j(D)\\ u\in M}} \lambda_M^{(2)} \\
&\quad + \lambda^{(1)}_{\{q_1,q_2\}} \mathds{I}_{12} 
+ \left( \tfrac{2}{\sqrt{5}} \lambda^{(1)}_{\{u,p\}} + \lambda^{(2)}_{\{q_1,q_2\}} \right) \mathds{I}_{21}
- \lambda^z_P \lambda_Q^r.
\end{align*}
\item  If  \( P = \{ u, v \}\in\mathcal{M}_2(D) \) and \( Q = \{ u, v, q \}\in\mathcal{M}_3(D) \):
\begin{eqnarray*}
\Sigma_{P,Q}^{z,r}&=&\tfrac{4}{5}(\lambda_{P}^{(1)}\mathds{I}_{21}+\lambda^{(2)}_{P}\mathds{I}_{11})+
 \big(\tfrac{2\sqrt{5}}{7}(\lambda^{(2)}_{\{v,q\}}+\lambda^{(2)}_{\{u,q\}})+\tfrac{20}{49}(\lambda_P^{(2)}+\lambda_Q^{(2)})\big)\mathds{I}_{22}\\&&+
 \big(\tfrac{2}{\sqrt{5}}(\lambda^{(1)}_{\{v,q\}}+\lambda^{(1)}_{\{u,q\}})+\tfrac{4}{5}(\lambda_P^{(1)}+\lambda_{Q}^{(1)}\big)\mathds{I}_{12}-\lambda_P^z\lambda_Q^r.\end{eqnarray*}
\item If \( P=\{p_1,p_2\}\in\mathcal{M}_2(D) \) and  \(Q=P\cup \{q_1,q_2\}=D \):
\begin{eqnarray*}
\Sigma_{P,Q}^{z,r} &=&\tfrac{4}{5}\lambda^{(2)}_{P} \mathds{I}_{11}+
  \sum\limits_{x\in\{p_3,p_4,P,D\}}\big(\tfrac{20}{49}\mathds{I}_{22}\lambda^{(2)}_{\{p_1,p_2\}\cup x}+\tfrac{4}{5}\mathds{I}_{12}\lambda^{(1)}_{\{p_1,p_2\}\cup x}\big)+
  \tfrac{4}{5}\lambda^{(1)}_{P}\mathds{I}_{21}\\&&+\sum\limits_{x\in\{p_3,p_4\}}\big(\tfrac{2\sqrt{5}}{7}(\lambda^{(2)}_{\{p_1,x\}}+\lambda^{(2)}_{\{p_2,x\}})\mathds{I}_{22} +\tfrac{2}{\sqrt{5}}(\lambda^{(1)}_{\{p_1,x\}}+\lambda^{(1)}_{\{p_2,x\}})\mathds{I}_{12}\big) \\&& +\lambda^{(2)}_{\{q_1,q_2\}}(\mathds{I}_{21}+\mathds{I}_{22})+\sum\limits_{x\in\{p_1,p_2\}}\big(\tfrac{2\sqrt{5}}{7}\mathds{I}_{22}\lambda^{(2)}_{\{x,q_1,q_2\}}+\tfrac{2}{\sqrt{5}}\mathds{I}_{12}\lambda^{(1)}_{\{x,q_1,q_2\}}) \big)\\&&+\lambda^{(1)}_{\{q_1,q_2\}} (\mathds{I}_{11}+\mathds{I}_{12})-\lambda_P^z\lambda^r_Q.
\end{eqnarray*}
\item For \(P, Q \in \)  \(\mathcal{M}_3(D)\) distinct, they have exactly two common elements. Assume \(P = \{u, v, p\}\) and \(Q = \{u, v, q\}\): 
\begin{eqnarray*}
\Sigma_{P,Q}^{z,r}&=& \lambda_{\{p,q\}}^{(1)}    \mathds{I}_{11}+\sum\limits_{x\in\{p,q,\{p,q\}\}}\big(\tfrac{2\sqrt{7}}{{5}}(\lambda_{\{u\}\cup x}^{(2)}+\lambda_{\{v\}\cup x}^{(2)}) +\tfrac{20}{49}\lambda_{\{u,v\}\cup x}^{(2)}\big)\mathds{I}_{22} \\&&+\tfrac{20}{49}\lambda_{\{u,v\}}^{(2)}\mathds{I}_{22}+\lambda_{\{p,q\}}^{(2)}\mathds{I}_{22}+\big(
  \tfrac{4}{5}(\lambda_{\{u,v\}}^{(1)}+\lambda_{\{u,v,q\}}^{(1)})+\tfrac{2}{\sqrt{5}}(\lambda_{\{u,q\}}^{(1)}+\lambda_{\{v,q\}}^{(1)})\big)\mathds{I}_{12}  \\&& +\big(
  \tfrac{4}{5}(\lambda_{\{u,v\}}^{(1)}+\lambda_{\{u,v,p\}}^{(1)})+\tfrac{2}{\sqrt{5}}(\lambda_{\{v,p\}}^{(1)}+\lambda_{\{u,p\}}^{(1)})\big)\mathds{I}_{21}-\lambda_P^z\lambda_Q^r. \end{eqnarray*}
\item If \(P=\{p_1,p_2,p_3\}\in\mathcal{M}_3(D)\) and \(Q=P\cup \{q\}=D\):
\begin{eqnarray*}\Sigma_{P,Q}^{z,r}&=&\tfrac{4}{5}\sum\limits_{x,y\in P}\lambda^{(2)}_{\{x,y,q\}}\mathds{I}_{11}+\tfrac{8}{5\sqrt{5}}((\mathds{I}_{21}+\mathds{I}_{12})\lambda_{D}^{(1)}+\lambda_D^{(2)}\mathds{I}_{11})\\&&+\tfrac{2}{\sqrt{5}}\sum\limits_{x\in P}(\lambda_{\{x,q\}}^{(2)}\mathds{I}_{11}+\lambda_{\{x,q\}}^{(1)}\mathds{I}_{21})+\tfrac{2\sqrt{5}}{{7}}\sum\limits_{x\in P}(\lambda_{\{x,q\}}^{(2)}\mathds{I}_{22} +\lambda_{\{x,q\}}^{(1)}\mathds{I}_{12} )\\&&+\sum\limits_{\underset{q\in M}{M\in\mathcal{M}_3(D)}}\big(\tfrac{20}{49}\mathds{I}_{22}\lambda_{M}^{(2)}+\tfrac{4}{5}\mathds{I}_{21}\lambda_{M}^{(1)}\big)+\sum\limits_{M\in\mathcal{M}_2(P)}(\tfrac{20}{49}\mathds{I}_{22}\lambda^{(2)}_M+\tfrac{4}{5}\mathds{I}_{12}\lambda^{(1)}_M)\\&&+\tfrac{40\sqrt{5}}{343}(\lambda_{P}^{(2)}+\lambda_{D}^{(2)})\mathds{I}_{22}+\tfrac{8}{5\sqrt{5}}\lambda_{P}^{(1)}\mathds{I}_{12}-\lambda_P^z\lambda_Q^r. \end{eqnarray*}
\end{enumerate}
\end{corollary}
\begin{remark}\label{test}
 To test the hypothesis \( H_0: \lambda_M^{(2)} = 0 \) for all \( M \in \mathcal{M}_j(D) \) with \( j = 2,\cdots, d \), we consider 
\[
T_n = n (\hat{\Lambda}^{(2)}-\Lambda^{(2)})' \Sigma_{2}^{-1}( \hat{\Lambda}^{(2)}-\Lambda^{(2)}),
\]
where \( \hat{\Lambda}^{(2)} \) is the vector of estimators for \( k = 2 \), and \( \Sigma_{2} \) is the submatrix of \( \Sigma \) corresponding to these parameters. Under \( H_0 \) and by Theorem \ref{cltlegendreextended}, 
\[
T_n \Rightarrow \chi^2(2^d-d-1),\quad \text{ as } n\rightarrow\infty.
\]  
For a given significance level \(\alpha,\)  we reject \(H_0\) if  \[1-\mathbb{P}[\chi^2(2^d-d-1)\leq T_n]\leq \alpha.\]
\end{remark}
\section{Simulation Study}\label{simulationstudy}
\subsection{Data generation}
We consider the copula \(C\) defined in \eqref{copfunction}. Given a vector \(\textbf{u}_\ell = (u_1, \dots, u_\ell)\) with \(2\leq \ell \leq d\), we define its marginal distribution as  
\[
C_\ell(\textbf{u}_\ell) = C(u_1, \dots, u_\ell, 1, \dots, 1),
\]  
with the corresponding density function  
\(
c_\ell(\textbf{u}_\ell).
\) 
The conditional distribution of \(U_\ell\) given \((U_1, \dots, U_{\ell-1})\), denoted by \(C_{\ell | 1, \dots, \ell-1}\), is given by:\begin{eqnarray*}
   C_{\ell|1\cdots \ell-1}(u_\ell|u_1,\cdots, u_{\ell-1})&=& P[U_\ell\leq u_\ell|U_1=u_1,\cdots,U_{\ell-1}=u_{\ell-1}]\\&=&\tfrac{\tfrac{\partial^{\ell-1} }{\partial u_1\cdots\partial u_{\ell-1}}C(u_1,\cdots,u_{\ell-1},u_\ell)}{\tfrac{\partial^{\ell-1} }{\partial u_1\cdots\partial u_{\ell-1}}C(u_1,\cdots,u_{\ell-1})}\\
    &=&\tfrac{\int_0^{u_\ell}c_\ell(u_1,\cdots,u_{\ell-1},u)du}{c_{\ell-1}(\textbf{u}_{\ell-1})}.
\end{eqnarray*}
Observe that\begin{eqnarray*}
\tfrac{c_\ell(\textbf{u}_{\ell})}{c_{\ell-1}(\textbf{u}_{\ell-1})}&=&1+\sum\limits_{k=1}^2\phi_k(u_\ell)
\Delta^{(k)}_\ell
\end{eqnarray*}  where \begin{equation}\label{delta}\Delta^{(k)}_\ell=\tfrac{ \sum\limits_{j=2}^{\ell}~\sum\limits_{\substack{M\in \mathcal{M}_j(\{1,\dots,\ell\}) \\ \ell\in M}} \lambda_M^{(k)} \prod\limits_{m\in M\setminus\{\ell\}} \phi_k(u_m)}{c_{\ell-1}(\textbf{u}_{\ell-1})}.\end{equation}
This leads to the following expression:\begin{eqnarray}\label{conditionaldistr}
C_{\ell|1\cdots\ell-1}(u_\ell|u_1,\cdots, u_{\ell-1})&=&u_\ell+\sum\limits_{k=1}^2\Delta^{(k)}_\ell \int_0^{u_\ell} \phi_k(u)du\nonumber\\
&=&2\sqrt{5}\Delta^{(2)}_\ell u_\ell^3+(\sqrt{3}\Delta^{(1)}_\ell-3\sqrt{5}\Delta^{(2)}_\ell)u_\ell^2\nonumber\\&&+(1-\sqrt{3}\Delta^{(1)}_\ell+\sqrt{5}\Delta^{(2)}_\ell)u_\ell.\end{eqnarray}

 To generate a random vector \((u_1,\dots,u_d)\) that follows the extended FGM copula defined in \eqref{copfunction}, proceed as follows:
\begin{enumerate}
    \item Generate a vector \((v_1,\dots,v_d)\) from the uniform distribution on \([0,1]^d\).
    \item Set \(u_1 =v_1\).
    \item For \(\ell = 2, \dots, d\), we use the \texttt{uniroot} function on \texttt{R} to find \(u_\ell\) as the unique solution on \([0,1]\) of the equation 
    \[
   2\sqrt{5}\Delta^{(2)}_\ell u_\ell^3+(\sqrt{3}\Delta^{(1)}_\ell-3\sqrt{5}\Delta^{(2)}_\ell)u_\ell^2+(1-\sqrt{3}\Delta^{(1)}_\ell+\sqrt{5}\Delta^{(2)}_\ell)u_\ell-v_\ell=0.
    \]
\end{enumerate}
In this way, the vector \((u_1,\dots,u_d)\) follows the copula \(C\) defined in \eqref{copfunction
}.
\begin{remark}\label{transformR}
Let \(\mathbf X=(X_1,\dots,X_4)\) follow the extended \(4\)-variate FGM copula with marginal c.d.f.s \(F_1,\dots,F_4\).  
For an observation \(\mathbf x\) define the marginally transformed vector \(\mathbf u\) by \(u_j = F_j(x_j)\).

The Rosenblatt transform (Rosenblatt 1952 \cite{rosenblatt1952remarks}) maps \(\mathbf u\) to \(\mathbf R=(R_1,\dots,R_4)\), where \(R_1=u_1\) and, for \(2\le\ell\le4\),
\[
R_{\ell}
=2\sqrt5\,\Delta^{(2)}_{\ell}\,u_{\ell}^{3}
+\bigl(\sqrt3\,\Delta^{(1)}_{\ell}-3\sqrt5\,\Delta^{(2)}_{\ell}\bigr)u_{\ell}^{2}
+\bigl(1-\sqrt3\,\Delta^{(1)}_{\ell}+\sqrt5\,\Delta^{(2)}_{\ell}\bigr)u_{\ell},
\]
with conditional coefficients \(\Delta^{(k)}_{\ell}\) as in~\eqref{delta}.  
\end{remark} 
Under the correct copula, each \(R_{\ell}\) is \(\mathrm{U}(0,1)\); thus Q–Q plots and Kolmogorov–Smirnov tests based on \(\mathbf R\) provide a direct goodness‑of‑fit check.
\subsection{Tables and comments}
Let  \(d=4\) and \(\Lambda=(0.05,-0.05,0.05,-0.05,0.05,-0.05,0.05,-0.05,0.05,-0.05,0.02,\\\textcolor{white}{e}\hspace{3.8cm}-0.05,0.05,-0.05,0.05,-0.05,0.05,-0.05,0.05,-0.05,0.05,-0.025)
\).\\

We generate an i.i.d. sample using the copula defined in equation \eqref{copfunction} with the algorithm outlined in the previous section. The estimated parameters, calculated using the estimator in \eqref{estimate}, are presented in Table \ref{simul2}. 

\begin{table}
\centering
\begin{tabular}{cc}
\begin{tabular}{lccc}
  \toprule
 \(\hat\Lambda^{(1)}\)& n=100 & n=1000 & n=10000 \\
  \midrule
$\hat\lambda^{(1)}_{\{1,2\}}$ & 0.027 & 0.051 & 0.061 \\ 
$\hat\lambda^{(1)}_{\{1,3\}}$ & -0.174 & -0.042 & -0.055 \\ 
$\hat\lambda^{(1)}_{\{1,4\}}$ & 0.035 & 0.027 & 0.05 \\ 
$\hat\lambda^{(1)}_{\{2,3\}}$ & 0.024 & -0.023 & -0.049 \\ 
$\hat\lambda^{(1)}_{\{2,4\}}$ & 0.014 & -0.018 & 0.045 \\ 
$\hat\lambda^{(1)}_{\{3,4\}}$ & -0.173 & -0.017 & -0.038 \\ 
$\hat\lambda^{(1)}_{\{1,2,3\}}$ & 0.057 & 0.042 & 0.061 \\ 
$\hat\lambda^{(1)}_{\{1,2,4\}}$ & -0.065 & -0.021 & -0.053 \\ 
$\hat\lambda^{(1)}_{\{1,3,4\}}$ & 0.259 & 0.042 & 0.047 \\ 
$\hat\lambda^{(1)}_{\{2,3,4\}}$ & -0.063 & -0.056 & -0.042 \\ 
$\hat\lambda^{(1)}_{\{1,2,3,4\}}$ & 0.015 & 0.023 & 0.022 \\ 
\toprule
\end{tabular}
&

\begin{tabular}{lccc}
  \toprule
 \(\hat\Lambda^{(2)}\)& n=100 & n=1000 & n=10000 \\
  \midrule
$\hat\lambda^{(2)}_{\{1,2\}}$ & -0.209 & -0.024 & -0.039 \\ 
$\hat\lambda^{(2)}_{\{1,3\}}$ & 0.186 & 0.109 & 0.044 \\ 
$\hat\lambda^{(2)}_{\{1,4\}}$ & 0.162 & -0.059 & -0.038 \\ 
$\hat\lambda^{(2)}_{\{2,3\}}$ & -0.124 & 0.017 & 0.062 \\ 
$\hat\lambda^{(2)}_{\{2,4\}}$ & -0.118 & -0.065 & -0.072 \\ 
$\hat\lambda^{(2)}_{\{3,4\}}$ & 0.22 & 0.106 & 0.051 \\ 
$\hat\lambda^{(2)}_{\{1,2,3\}}$ & -0.036 & -0.056 & -0.054 \\ 
$\hat\lambda^{(2)}_{\{1,2,4\}}$ & -0.109 & 0.046 & 0.052 \\ 
$\hat\lambda^{(2)}_{\{1,3,4\}}$ & 0.026 & -0.106 & -0.049 \\ 
$\hat\lambda^{(2)}_{\{2,3,4\}}$ & 0.065 & 0.035 & 0.046 \\ 
$\hat\lambda^{(2)}_{\{1,2,3,4\}}$ & -0.053 & -0.043 & -0.033 \\ 
  \toprule
\end{tabular}
\end{tabular}
\caption{Estimates for parameters \(\lambda_M^{(1)}\) and \(\lambda_M^{(2)}\)}
\label{simul2}
\end{table}
As the sample size increases from \(n=100\) to \(n=10,000\), the estimates become more stable and consistent, with initial variability in smaller samples diminishing and the results better reflecting the theoretical values. 
\begin{table}
\centering
\begin{tabular}{c c}
\begin{tabular}{lr}
  \hline
  Est. & C.P \\ 
  \hline
  $\hat\lambda^{(1)}_{\{1,2\}}$ & 95.7 \\ 
  $\hat\lambda^{(1)}_{\{1,3\}}$ & 94.8 \\ 
  $\hat\lambda^{(1)}_{\{1,4\}}$ & 94.8 \\ 
  $\hat\lambda^{(1)}_{\{2,3\}}$ & 94.3 \\ 
  $\hat\lambda^{(1)}_{\{2,4\}}$ & 94.8 \\ 
  $\hat\lambda^{(1)}_{\{3,4\}}$ & 95.5 \\ 
  $\hat\lambda^{(1)}_{\{1,2,3\}}$ & 97.2 \\ 
  $\hat\lambda^{(1)}_{\{1,2,4\}}$ & 96.7 \\ 
  $\hat\lambda^{(1)}_{\{1,3,4\}}$ & 94.2 \\ 
  $\hat\lambda^{(1)}_{\{2,3,4\}}$ & 96.1 \\ 
  $\hat\lambda^{(1)}_{\{1,2,3,4\}}$ & 96.5 \\ 
  \hline
  \multicolumn{2}{c}{(a)} \\
\end{tabular}
&
\begin{tabular}{lr}
  \hline
 Est. & C.P \\ 
  \hline
  $\hat\lambda^{(2)}_{\{1,2\}}$ & 94.0 \\ 
  $\hat\lambda^{(2)}_{\{1,3\}}$ & 94.9 \\ 
  $\hat\lambda^{(2)}_{\{1,4\}}$ & 94.6 \\ 
  $\hat\lambda^{(2)}_{\{2,3\}}$ & 95.5 \\ 
  $\hat\lambda^{(2)}_{\{2,4\}}$ & 93.9 \\ 
  $\hat\lambda^{(2)}_{\{3,4\}}$ & 94.1 \\ 
  $\hat\lambda^{(2)}_{\{1,2,3\}}$ & 94.3 \\ 
  $\hat\lambda^{(2)}_{\{1,2,4\}}$ & 94.7 \\ 
  $\hat\lambda^{(2)}_{\{1,3,4\}}$ & 94.5 \\ 
  $\hat\lambda^{(2)}_{\{2,3,4\}}$ & 94.0 \\ 
  $\hat\lambda^{(2)}_{\{1,2,3,4\}}$ & 95.0 \\ 
  \hline
  \multicolumn{2}{c}{(b)} \\
\end{tabular}
\end{tabular}
\caption{Coverage probabilities for $\hat\lambda^{(1)}_M$ and $\hat\lambda^{(2)}_M$ with $n = 1000$ over 1000 iterations.}
\label{coverageprobtab}
\end{table}
Table \ref{coverageprobtab} presents the estimated coverage probabilities obtained from 1000 iterations using a sample size of 1000 at a nominal 95\% confidence level. The results clearly indicate that the empirical coverage probabilities are very close to the targeted 95\%.

\section{Case Study: Application to Bearing Data}\label{casestudy}
For this case study, we utilize the IMS Bearing Data Set provided by the NSF I/UCR Center for Intelligent Maintenance Systems with support from Rexnord Corp. (see Lee et al., 2007 \cite{Lee2007}). The data comprises three separate test-to-failure experiments conducted using a test rig equipped with four Rexnord ZA-2115 double row bearings mounted on a shaft driven at a constant speed of 2000 RPM and subjected to a 6000 lb radial load. Vibration data were recorded using high-sensitivity PCB 353B33 ICP accelerometers placed on the bearing housings, two per bearing for the first dataset and one per bearing for the second and third. Each dataset consists of thousands of 1-second vibration signal snapshots (20,480 data points per file, sampled at 20 kHz), collected at regular 10-minute intervals. All bearings were force-lubricated, and failures occurred only after exceeding their designed lifetime. The setup suggests a plausible dependence among the bearings, as they are mounted on a shared shaft and subjected to the same operating conditions.

\begin{figure}[th]
    \centering
    \includegraphics[width=9.4cm, height=7cm]{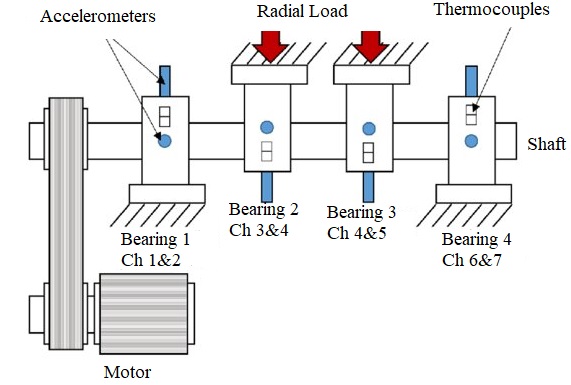}
    \caption{Test rig illustration, adapted from Ota and Kimura (2021) \cite{OtaKimura2021} and Lee et al. (2007) \cite{Lee2007}}
    \label{fig1}
\end{figure}
\subsection{Parameters Estimation}
In this study, we construct a 4-variate dataset by extracting one observation per bearing from each of the 2,156 files in Data Set 1, using channels 1, 3, 5, and 7 corresponding to Bearings 1 through 4, respectively.  Each observation corresponds to the first recorded vibration signal in a given file, resulting in a 4-dimensional vector at each time point. This dataset is defined as \( D = \{(x_{i1}, x_{i2}, x_{i3}, x_{i4}) : i = 1, \ldots, 2156\} \), where \( x_{ij} \) denotes the \(i^{th}\) vibration signal from the \( j^{th}\) channel.

Ota and Kimura (2021)\cite{OtaKimura2021} analyzed this construction using a 4-variate FGM copula. To model the marginal distributions, they applied the Kolmogorov–Smirnov test and identified the generalized t-distribution as the best fit, with density
\[
f_\delta(x) = \frac{\left[1 + \left(\frac{x - a}{b\sqrt{c}}\right)^2\right]^{-(c+1)/2}}{b\sqrt{c} \, B(c/2, 1/2)},
\]

Here, \( \delta = (a, b, c) \), and \( B \) denotes the beta function. The components of \( \delta \) for each channel were estimated {via maximum likelihood} using the full dataset \( D \), with the results reported in Table \ref{CDFestimates} below. These estimates are treated as known for our analysis.

\begin{table}[h!]
    \centering
    \renewcommand{\arraystretch}{1.2}
    \begin{tabular}{ccccccc}
    \toprule
          \(\delta\) & Ch 1 & Ch 3 & Ch 5 & Ch 7 \\
        \midrule
        \( {a} \) & -0.119 & -0.116 & -0.115 & -0.116 \\
         \( {b} \) & 0.0877 & 0.0905 & 0.103 & 0.0743 \\
     \( {c} \) & 16.0 & 26.8 & 8.28 & 4.73 \\
        \toprule
    \end{tabular}
    \caption{Estimates of \(\delta\) for different channels using D. Source: \cite{OtaKimura2021}}
    \label{CDFestimates}
\end{table}
We divide the dataset into three subsets: Data 1 consists of observations 1 to 1100, Data 2 includes observations 1101 to 2156, and Data 3 comprises the full dataset. Tables \ref{lambda2real} and \ref{lambda1real} present the estimated copula parameters from the extended 4-variate FGM copula model defined in \eqref{ExtendedDensity}, obtained using the estimator in \eqref{estimate}. The tables also report two-sided \(p\)-values for testing the null hypotheses \(H_0: \lambda^{(k)}_M = 0\) for each parameter, based on Data 3. In addition, Table \ref{lambda2real} includes the \(p\)-value for testing the hypothesis \(H_0: \Lambda^{(2)} = \mathbf{0}\), as described in Remark \ref{test}.

\begin{table}[ht]
\centering
\begin{tabular}{lccccccc}
  \toprule
  & \(\hat\lambda^{(1)}_{\{1,2\}}\) & \(\hat\lambda^{(1)}_{\{1,3\}}\) & \(\hat\lambda^{(1)}_{\{1,4\}}\) & \(\hat\lambda^{(1)}_{\{2,3\}}\) & \(\hat\lambda^{(1)}_{\{2,4\}}\) & \(\hat\lambda^{(1)}_{\{3,4\}}\) & \(\hat\lambda^{(1)}_{\{1,2,3\}}\) \\  
  \midrule
Data 1 & 0.160 & -0.030 & 0.038 & -0.150 & -0.120 & 0.090 & -0.001 \\  
Data 2 & 0.150 & -0.062 & -0.027 & -0.130 & -0.180 & 0.093 & -0.022 \\  
Data 3 & 0.156 & -0.046 & 0.006 & -0.138 & -0.153 & 0.091 & 0.010 \\  
\(p\)-value (Data 3) & 0.000 & 0.030 & 0.785 & 0.000 & 0.000 & 0.000 & 0.631 \\  
  \bottomrule
\end{tabular}

\vspace{0.5cm} 

\begin{tabular}{lccccc}
  \toprule
  & \(\hat\lambda^{(1)}_{\{1,2,4\}}\) & \(\hat\lambda^{(1)}_{\{1,3,4\}}\) & \(\hat\lambda^{(1)}_{\{2,3,4\}}\) & \(\hat\lambda^{(1)}_{\{1,2,3,4\}}\) \\  
  \midrule
Data 1 & -0.044 & -0.004 & -0.040 & 0.043 \\  
Data 2 & 0.031 & -0.040 & 0.034 & -0.028 \\  
Data 3 & -0.007 & -0.020 & -0.001 & 0.008 \\  
\(p\)-value (Data 3) & 0.736 & 0.359 & 0.952 & 0.685 \\  
  \bottomrule
\end{tabular}

\caption{Copula parameter estimates for \(\lambda_M^{(1)}\).}
\label{lambda1real}
\end{table}

\begin{table}[ht]
\centering
\begin{tabular}{lccccccl}
  \toprule
  & \(\hat\lambda^{(2)}_{\{1,2\}}\) & \(\hat\lambda^{(2)}_{\{1,3\}}\) & \(\hat\lambda^{(2)}_{\{1,4\}}\) & \(\hat\lambda^{(2)}_{\{2,3\}}\) & \(\hat\lambda^{(2)}_{\{2,4\}}\) & \(\hat\lambda^{(2)}_{\{3,4\}}\) & \(\hat\lambda^{(2)}_{\{1,2,3\}}\) \\  
  \midrule
Data 1 & -0.041 & 0.028 & 0.046 & 0.052 & 0.030 & 0.054 & 0.038 \\  
Data 2 & -0.051 & 0.079 & 0.047 & 0.040 & 0.012 & 0.045 & -0.036 \\  
Data 3 & -0.046 & 0.053 & 0.046 & 0.046 & 0.021 & 0.050 & 0.002 \\  
\(p\)-value (Data 3) & 0.030 & 0.015 & 0.033 & 0.034 & 0.33 & 0.022 & 0.939 \\  
  \bottomrule
\end{tabular}

\vspace{0.5cm} 

\begin{tabular}{lccccc}
  \toprule
  & \(\hat\lambda^{(2)}_{\{1,2,4\}}\) & \(\hat\lambda^{(2)}_{\{1,3,4\}}\) & \(\hat\lambda^{(2)}_{\{2,3,4\}}\) & \(\hat\lambda^{(2)}_{\{1,2,3,4\}}\) & \(H_0: \hat\Lambda^{(2)}=0\) vs \(H_1: \hat\Lambda^{(2)}\neq \mathbf{0}\)\\  
  \midrule
Data 1 & -0.056 & 0.021 & 0.009 & 0.022 &  \\  
Data 2 & -0.027 & 0.036 & 0.026 & -0.035 & \\  
Data 3 & -0.042 & 0.028 & 0.017 & -0.006 & \\  
\(p\)-value (Data 3) & 0.053 & 0.201 & 0.437 & 0.795 &0.005 \\  
  \bottomrule
\end{tabular}

\caption{Copula parameter estimates for \(\lambda_M^{(2)}\).}
\label{lambda2real}
\end{table}

The results in Tables \ref{lambda1real} and \ref{lambda2real} point out two important characteristics. First, the estimates of pairwise dependence parameters remain stable across the three temporal subsets, suggesting a persistent and consistent dependence structure over time. Second, only the pairwise interactions are statistically significant at the \(\alpha = 5\%\) level, except for \(\hat\lambda_{\{1,4\}}^{(1)}\) and \(\hat\lambda_{\{2,4\}}^{(2)}\), which do not show statistical significance. The stronger significance observed for adjacent pairs further emphasizes that the dependence structure is primarily local. The sign pattern observed in Table \ref{lambda1real} aligns with the results reported by Ota and Kimura (2021) \cite{OtaKimura2021}: positive association appears between bearings 1–2 and 3–4, negative association between 1–3, 2–3, and 2–4, while bearings 1–4 remain nearly independent. In contrast, all parameters corresponding to interactions involving three or more bearings are essentially zero, indicating that higher-order dependencies do not contribute meaningfully beyond the pairwise structure.

Table \ref{lambda2real} shows that, beyond the structure captured by the classical \(4\)-variate FGM copula, a residual dependence is present and accounted for by the second Legendre term in the extended model. While the classical parameters generally display stronger individual significance than the extended ones, the test of \(H_0: \Lambda^{(2)} = \mathbf{0}\) is rejected with a \(p\)-value of 0.005, confirming the contribution of the second-order component. A particularly relevant case is the pair \(\{1,4\}\): the corresponding parameter in the extended model,  \(\lambda_{1,4}^{(2)}\) is significant, suggesting that the assumption of near independence stated by only looking at \(\lambda_{1,4}^{(1)}\), between these bearings no longer holds . Furthermore, several variable pairs that already exhibit strong first-order dependence in Table \ref{lambda1real} also show significant second-order terms, indicating that the \(4\)-variate FGM copula alone is insufficient. \\
\subsection{Model Selection}
This analysis indicates that the density of a more suitable copula model for this dataset is \begin{equation}\label{suitable}
c(u_1,u_2,u_3,u_4)=1+\sum_{k=1}^{2}\,\mathbf v_k^\top {\Lambda^{(k)}}^* ,
\end{equation}
with
\[
{\Lambda^{(1)}}^*=(0.156,\,-0.046,\,0,\,-0.138,\,-0.153,\,0.091,\,0,\,0,\,0,\,0,\,0)\]
\[{\Lambda^{(2)}}^*=(-0.046,\,0.053,\,0.046,\,0.046,\,0,\,0.050,\,0,\,0,\,0,\,0,\,0).
\]
Zeros in these vectors are obtained by setting every coefficient whose two‑sided $p$‑value exceeds $0.05$ to zero.

On the other hand, the full 4‑variate FGM copula  is obtained by considering 
\[
{\Lambda^{(1)}}^{*}=(0.156,\,-0.046,\,0.006,\,-0.138,\,-0.153,\,0.091,\,0.010,\,-0.007,\,-0.020,\,-0.001,\,0.008),\] and \({\Lambda^{(2)}}^{*}=\mathbf 0
\).

Table \ref{AICcomparison} reports the Akaike and Bayesian Information Criteria.  
The extended copula achieves lower values for both AIC and BIC, indicating a better fit for the data.

\begin{table}[ht]
\centering
\begin{tabular}{lrr}
\toprule
Model & AIC & BIC\\
\midrule
\(4\)-variate FGM (\ref{classicalFGM})          & $-132.72$ & $-70.29$\\
\(4\)-varaite Extended FGM \ref{thedvariateFGMcopula} & $-161.55$ & $-106.79$\\
\bottomrule
\end{tabular}
\caption{Information‑criterion comparison for the two copula specifications.}
\label{AICcomparison}
\end{table}

Figure \ref{deviation} shows the deviation of the Rosenblatt‑transformed components (see Remark \ref{transformR}) to the  $\mathcal{U}(0,1)$ distribution, with the corresponding Kolmogorov–Smirnov \(p\)-values displayed in each panel.  
The large \(p\)-values indicate that the null hypothesis of uniformity cannot be rejected, confirming the fitted copula's adequacy.


\begin{figure}[htbp]
    \centering
    \includegraphics[width=17cm, height=10cm]{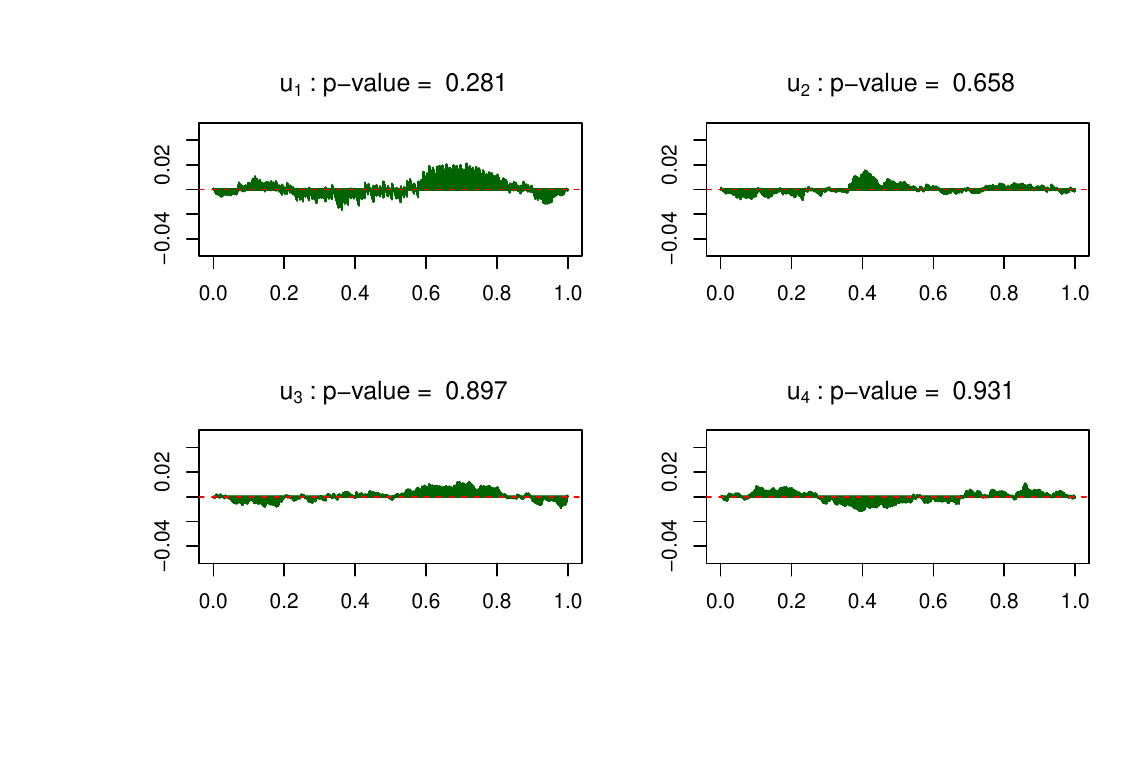}
    \caption{Deviation-from-uniformity plot for Rosenblatt-transformed variables. The horizontal line at zero indicates a perfect uniform fit.}
    \label{deviation}
\end{figure}
\section{Conclusion}
This paper develops an extended \(
d\)-variate FGM copula model based on Legendre polynomials and proposes explicit estimators for its parameters under the i.i.d setting. The estimators are shown to be unbiased, consistent, and asymptotically normal. A simulation study confirmed their finite-sample accuracy. Applying the model to the Bearing dataset revealed that the classical 4-variate FGM copula fails to capture certain residual dependencies. In contrast, the extended version, especially in its reduced form, demonstrates improved fit as supported by significance tests, Rosenblatt diagnostics, and information criteria.

\end{document}